\documentclass{article}
\usepackage[dvipsnames]{xcolor}
\usepackage{gastex}
\usepackage{amsmath}
\usepackage{amssymb}
\usepackage{upref}
\usepackage{multirow}
\usepackage{multicol}
\usepackage{url}
\usepackage{makeidx}
\usepackage{theorem}
\newtheorem{theorem}{Theorem}
\newtheorem{lemma}[theorem]{Lemma}
\newtheorem{proposition}[theorem]{Proposition}
\newtheorem{corollary}[theorem]{Corollary}
{\theorembodyfont{\rmfamily}%
  \newtheorem{example}[theorem]{Example}
   }
\newenvironment{proof}{\noindent\textit{Proof.}}
{\QED\vskip\theorempostskipamount} 
\newenvironment{proofof}[1]{\noindent\textit{Proof
    \protect{#1}.}}
                       {\QED\vskip\theorempostskipamount}
\def\petitcarre{\vrule height4pt width 4pt depth0pt}
\def\QED{\relax\ifmmode\eqno{\hbox{\petitcarre}}\else{%
  \unskip\nobreak\hfil\penalty50\hskip2em\hbox{}\nobreak\hfil
  \petitcarre
  \parfillskip=0pt \finalhyphendemerits=0\par\smallskip}
  \fi}
\newcommand\RR{\mathcal{R}}

\def\un(#1){\underline{#1}\,}
\DeclareMathOperator{\Card}{Card}

\definecolor{ivoire}{rgb}{0.99,0.99,0.8}
\usepackage{calc}
\newcounter{hours}\newcounter{minutes}
\newcommand\computetime{\setcounter{hours}{\time/60}%
  \setcounter{minutes}{\time-\value{hours}*60}%
  \thehours\,h\,\theminutes}
\newcommand\dateandtime{\today\quad\computetime}
\numberwithin{theorem}{section}
\numberwithin{equation}{section}
\numberwithin{figure}{section}
\numberwithin{table}{section}
\title{The finite index basis property}
\author{Val\'erie Berth\'e$^1$, Clelia De Felice$^2$, 
Francesco Dolce$^3$, Julien Leroy$^4$,\\
 Dominique Perrin$^3$,
Christophe  Reutenauer$^5$,
Giuseppina Rindone$^3$\\\\
$^1$CNRS, Universit\'e Paris 7,
$^2$Universit\`a degli Studi di Salerno,\\
$^3$Universit\'e Paris Est, LIGM,
$^4$Universit\'e du Luxembourg,\\
 $^5$Universit\'e du Qu\'ebec \`a Montr\'eal}
\date{\dateandtime}
\begin{document}
%========================
%========== Lists ==================================
\makeatletter
\def\@listI{%
  \leftmargin\leftmargini
  \setlength{\parsep}{0pt plus 1pt minus 1pt}
  \setlength{\topsep}{2pt plus 1pt minus 1pt}
  \setlength{\itemsep}{0pt}
}
\let\@listi\@listI
\@listi
\def\@listii {%
  \leftmargin\leftmarginii
  \labelwidth\leftmarginii
  \advance\labelwidth-\labelsep
  \setlength{\topsep}{0pt plus 1pt minus 1pt}
}
\def\@listiii{%
  \leftmargin\leftmarginiii
  \labelwidth\leftmarginiii
  \advance\labelwidth-\labelsep
  \setlength{\topsep}{0pt plus 1pt minus 1pt}
%              \topsep    0\p@ \@plus\p@\@minus\p@
  \setlength{\parsep}{0pt} 
  \setlength{\partopsep}{1pt plus 0pt minus 1pt}
}
\makeatother
%====================
\maketitle
%========================

\begin{abstract}
We describe in this paper a connection between bifix codes, symbolic
dynamical systems
 and free groups. This is in the spirit of the
connection
established previously for the symbolic systems
corresponding to Sturmian words. We introduce a class of sets of
factors of an infinite word with linear factor complexity
containing Sturmian sets
and regular interval exchange sets, namely the class of tree
sets.
We prove as a main result that for a uniformly recurrent
tree  set $S$, 
 a finite bifix
code $X$ on the alphabet $A$ is $S$-maximal of $S$-degree $d$ if and only if it
is the basis of a subgroup of index $d$ of the free group on $A$.
\end{abstract}
\tableofcontents
\section{Introduction}

In this paper we study a relation between symbolic dynamical systems
 and bifix codes.
The paper is a continuation of the paper with
part of the present list of authors
on bifix codes and Sturmian words
\cite{BerstelDeFelicePerrinReutenauerRindone2012}.
We understand here by Sturmian words the generalization to arbitrary alphabets,
often called strict episturmian words or Arnoux-Rauzy words
(see  the survey~\cite{GlenJustin2009}), of the classical Sturmian
words
on two letters.

As a main result, we prove that, under  natural hypotheses
satisfied by  a Sturmian set
$S$,
 a finite bifix
code $X$ on the alphabet $A$ is $S$-maximal of $S$-degree $d$ if and only if it
is the basis of a subgroup of index $d$ of the free group on $A$
(Theorem~\ref{newTheoremBasis} called below the Finite Index Basis Theorem).

The proof  uses the property,
proved in~\cite{BertheDeFeliceDolceLeroyPerrinReutenauerRindone2013d},
that the sets of first return words in a uniformly recurrent tree set
containing the alphabet $A$
form a basis of the free group on $A$
(this result is referred to below as the Return Words Theorem).

We actually introduce several classes of uniformly recurrent
sets of words on $k+1$
letters having all $kn+1$ elements of length $n$ for all $n\ge 0$.

The smallest class ($BS$) is formed of the Sturmian sets on a binary
alphabet, that is, with $k=1$ (see Figure~\ref{drawingClasses}). It is contained both in the class of regular
interval exchange sets (denoted $RIE$) and of Sturmian sets (denoted
$S$). Moreover, it can be shown that the intersection of
$RIE$ and $S$ is reduced to $BS$. Indeed, Sturmian sets on more than two
letters are not the set of factors of an interval exchange transformation
with each interval labeled by a distinct letter (the construction
in~\cite{ArnouxRauzy1991} allows one to obtain the Sturmian sets of $3$
letters as an exchange of $7$ intervals labeled by $3$ letters).

The next one is the class of uniformly recurrent
sets satisfying the tree condition ($T$),
which contains the  previous ones.
The class of uniformly recurrent sets satisfying
the neutrality condition ($N$) contains
the class $T$. All these classes are contained in the
class of uniformly recurrent sets of complexity $kn+1$ on an alphabet with $k+1$ letters.

We have tried in all the paper to use the weakest possible
conditions to prove our results.
As an example, we prove that, under the neutrality
condition,  any finite $S$-maximal bifix code of $S$-degree
$d$ has $1+d(\Card(A)-1)$ elements (Theorem~\ref{theoremCardinality}
called below the Cardinality Theorem).

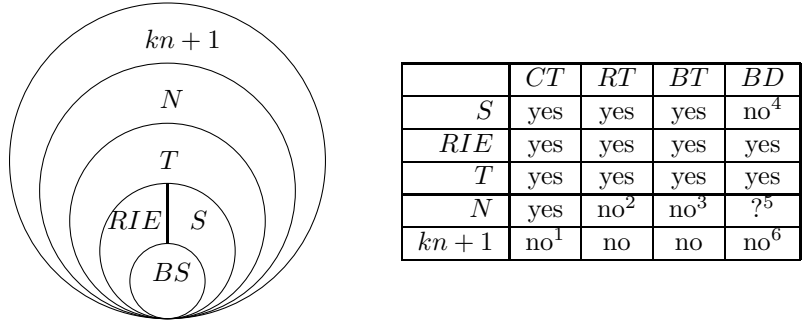
\begin{figure}[hbt]
\centering
\gasset{AHnb=0,Nadjust=wh}
\begin{picture}(120,50)
\put(0,0){\begin{picture}(60,45)

\drawcircle(30,5,10)\put(28,5){$BS$}

\drawcircle(30,9,18)%\drawline(30,10,30,20)
\put(30,10){\line(0,1){8}}
\put(22,12){$RIE$}\put(33,12){$S$}

\drawcircle(30,13,26)\put(29,20){$T$}

\drawcircle(30,17,34)\put(29,28){$N$}

\drawcircle(30,21,42)\put(27,36){$kn+1$}

\
\end{picture}
}
\put(60,20){
\begin{tabular}{|r|c|c|c|c|}\hline
     &$CT$&$RT$&$BT$&$BD$\\ \hline
$S$&yes&yes&yes&no$^4$\\ \hline
$RIE$&yes&yes&yes&yes\\ \hline
$T$&yes&yes&yes&yes\\ \hline
$N$&yes&no$^2$&no$^3$&?$^5$\\ \hline
$kn+1$&no$^1$&no&no&no$^6$\\ \hline
\end{tabular}
}
\end{picture}
\caption{The classes of uniformly recurrent sets on $k+1$ letters: Binary Sturmian ($BS$), Regular interval
  exchange ($RIE$),
Sturmian ($S$), Tree ($T$), Neutral ($N$), and finally of complexity
$kn+1$ (1: see Example 3.10 below, 2: see Example 5.9 in~\cite{BertheDeFeliceDolceLeroyPerrinReutenauerRindone2013d}, 
3: see Example~\ref{exampleJulien3} below,
4: see Example 4.4 in~\cite{BertheDeFeliceDolceLeroyPerrinReutenauerRindone2014},
5: it can be shown that the neutrality is preserved but it is not known
whether the uniform recurrence is,
6: see Example~\ref{exampleBifixeChacon2} below). }\label{drawingClasses}
\end{figure}

The class $RIE$ is closed under decoding by a maximal bifix
code (Theorem 3.13 in~\cite{BertheDeFeliceDolceLeroyPerrinReutenauerRindone2014} referred to as the
Bifix Decoding Theorem) but it is not
the case for Sturmian sets. In contrast, the uniformly recurrent
tree sets form a class of sets
containing the Sturmian sets and the regular interval exchange sets
which is closed under decoding by
a maximal bifix code (see~\cite{BertheDeFeliceDolceLeroyPerrinReutenauerRindone2013c}) 
and for which the Finite Index Basis Theorem is true.

 For each class, the array on the right of Figure~\ref{drawingClasses} 
 indicates
whether it satisfies the Cardinality Theorem ($CT$), the Return Words
Theorem ($RT$), the Finite Index Basis Theorem ($BT$)
or the Bifix Decoding Theorem ($BD$). All these classes
are distinct. 

The paper is organized as follows.
In Section~\ref{sectionNeutrality}, we introduce strong, weak and
neutral
sets. We prove the Cardinality Theorem in neutral sets
(Theorem~\ref{theoremCardinality}). We also prove a converse
 in the sense that a uniformly
recurrent set $S$ containing the alphabet and such that the
Cardinality Theorem holds for any finite $S$-maximal bifix code 
is  neutral (Theorem~\ref{theoremClelia}).

In Section~\ref{sectionTreePlanarTree}, we introduce 
acyclic and tree sets. The family of tree sets contains
 Sturmian sets and, as shown in~\cite{BertheDeFeliceDolceLeroyPerrinReutenauerRindone2014},  regular interval exchange sets.
We 
prove, as a main result, that  uniformly recurrent tree
 sets  satisfy the finite index property
(Theorem~\ref{newTheoremBasis}), a result
 which is proved in~\cite{BerstelDeFelicePerrinReutenauerRindone2012}
for a Sturmian set. The proof uses a result of~\cite{BertheDeFeliceDolceLeroyPerrinReutenauerRindone2013d}
concerning bifix codes in acyclic sets (Theorem 4.2 referred to as the
Saturation Theorem). It also uses the Return Words
Theorem proved in~\cite{BertheDeFeliceDolceLeroyPerrinReutenauerRindone2013d}.
We also prove a converse of Theorem~\ref{newTheoremBasis}, in the sense that
a uniformly recurrent set which has the finite index basis property
is a tree set (Corollary~\ref{corollaryConverseFiniteIndex}).
\paragraph{Ackowledgement} This work was supported by grants from
R\'egion \^{I}le-de-France, the ANR projects Eqinocs ANR-11-
BS02-004
and Dyna3S, ANR-13-BS02-003,
 the Labex Bezout,
the FARB Project
``Aspetti algebrici e computazionali nella teoria dei codici,
degli automi e dei linguaggi formali'' (University of Salerno, 2013)
and the MIUR PRIN 2010-2011 grant
``Automata and Formal Languages: Mathematical and Applicative Aspects''
H41J12000190001.
 We warmly thank the referee
for his useful remarks on the first version of the paper.

%%%%%%%%%%%%%%%%%%%%%%%%%%%%%%%%%%
\section{Preliminaries}
In this section, we first recall some definitions concerning words, prefix codes
and bifix codes.
We give the definitions of recurrent and uniformly recurrent
sets of words. We also give the definitions and basic
properties of bifix codes (see~\cite{BerstelDeFelicePerrinReutenauerRindone2012} for a more detailed
presentation).
%%%%%%%%%%%%%%%%%%%%
\subsection{Words}
In this section, we give definitions concerning extensions of words.
We define recurrent sets and sets of first return words.
For all undefined notions, we refer to~\cite{BerstelPerrinReutenauer2009}.
\subsubsection{Recurrent sets}
Let $A$ be a finite nonempty alphabet. All words considered below,
unless
stated explicitly, are supposed to be on the alphabet $A$.
We denote by $A^*$ the set of all words on $A$.
We denote by $1$ or by $\varepsilon$ the empty word. We refer to
\cite{BerstelPerrinReutenauer2009} for the notions of prefix, suffix,
factor of a word.

A set of words is said to be \emph{prefix-closed} (resp. \emph{factorial}) if it contains the
prefixes (resp. factors) of its elements.

Let $S$ be a set of words on the alphabet $A$.
 For $w\in S$,
we denote
\begin{eqnarray*}
L(w)&=&\{a\in A\mid aw\in S\},\\
R(w)&=&\{a\in A\mid wa\in S\},\\
E(w)&=&\{(a,b)\in A\times A\mid awb\in S\}
\end{eqnarray*}
and further
\begin{displaymath}
\ell(w)=\Card(L(w)),\quad r(w)=\Card(R(w)),\quad e(w)=\Card(E(w)).
\end{displaymath}
A word $w$ is \emph{right-extendable} if $r(w)>0$,
\emph{left-extendable} if $\ell(w)>0$ and \emph{biextendable} if
$e(w)>0$. A 
factorial set
$S$ is called \emph{right-extendable}
(resp. \emph{left-extendable}, resp. \emph{biextendable}) if every word in $S$ is
right-extendable (resp. left-extendable, resp. biextendable).

A word $w$ is called
\emph{right-special}
if $r(w)\ge 2$. It is called \emph{left-special} if $\ell(w)\ge 2$.
It is called \emph{bispecial} if it is both right and left-special.

A set of words $S\ne \{1\}$ is \emph{recurrent} if it is factorial and if for every
$u,w\in S$ there is a $v\in S$ such that $uvw\in S$. A recurrent set
 is biextendable.

A  set of words $S$ is said to be \emph{uniformly recurrent} if it is
right-extendable and if, for any word $u\in S$, there exists an integer $n\ge
1$
such that $u$ is a factor of every word of $S$ of length $n$.
A uniformly recurrent set is recurrent, and thus biextendable.

A \emph{morphism} $f:A^*\rightarrow B^*$ is a monoid morphism from
$A^*$ into $B^*$. If $a\in A$ is such that the word $f(a)$ begins with
$a$ and if $|f^n(a)|$ tends to infinity with $n$, there is a unique
infinite word denoted $f^\omega(a)$ which has all words $f^n(a)$
as prefixes. It is called a \emph{fixpoint} of the morphism $f$.

A morphism $f:A^*\rightarrow A^*$ is called \emph{primitive} if there
is an integer $k$ such that for all $a,b\in A$, the letter $b$
appears in $f^k(a)$. If $f$ is a primitive morphism, the set
of factors of any fixpoint
of $f$ is uniformly recurrent (see~\cite{PytheasFogg2002},
Proposition 1.2.3 for example).

A morphism $f:A^*\rightarrow B^*$ is \emph{trivial} if $f(a)=1$
for all $a\in A$. The image of a uniformly recurrent set
by a nontrivial morphism is uniformly recurrent (see
\cite{AlloucheShallit2003}, Theorem 10.8.6 and Exercise 10.11.38).

An infinite word is \emph{episturmian} if the set of its factors 
is closed under reversal and contains for each $n$ at most one word
of length $n$ which is right-special. It is a \emph{strict episturmian} word
if it has exactly one right-special word of each length and
moreover each right-special factor $u$ is such that $r(u)=\Card(A)$.

%When $\Card(A)=2$, one uses the term \emph{Sturmian} instead of strict
%episturmian (in this case all right-special words are strict).

A \emph{Sturmian set} is a set of words which 
is the set of factors of a strict episturmian word.
Any Sturmian set is uniformly recurrent (see~\cite{BerstelDeFelicePerrinReutenauerRindone2012}).
\begin{example}\label{exampleFibonacci}
Let $A=\{a,b\}$.
The Fibonacci word is the fixpoint $x=f^\omega(a)=abaababa\ldots$ of the
morphism $f:A^*\rightarrow A^*$ defined by $f(a)=ab$ and $f(b)=a$.
It is a Sturmian word (see~\cite{Lothaire2002}). The set $F(x)$ of factors of $x$ is the
\emph{Fibonacci set}.
\end{example}
\begin{example}\label{exampleTribonacci}
Let $A=\{a,b,c\}$.
The Tribonacci word is the fixpoint $x=f^\omega(a)=abacaba\cdots$ of the morphism
$f:A^*\rightarrow A^*$ defined by $f(a)=ab$, $f(b)=ac$, $f(c)=a$.
It is a strict episturmian word (see~\cite{JustinVuillon2000}).
The set $F(x)$ of factors of $x$ is the \emph{Tribonacci set}.
\end{example}

%%%%%%%%%%%%%%%%%%%%%%%%%%%%%%%%%%%%%%%%
\subsection{Bifix codes}
In this section, we present basic definitions concerning prefix codes
and bifix codes. For a more detailed presentation, 
see~\cite{BerstelPerrinReutenauer2009}. We also describe an
operation on bifix codes called internal 
transformation and prove a property
of this transformation (Proposition~\ref{propositionInternal}).
It will be used in Section~\ref{sectionConverseCardinality}.
\subsubsection{Prefix codes}
A \emph{prefix code} is a set of nonempty words which does not contain any
proper prefix of its elements. A suffix code is defined symmetrically.
A  \emph{bifix code} is a set which is both a prefix code and a suffix
code.

A \emph{coding morphism} for a prefix code $X\subset A^+$ is a morphism
$f:B^*\rightarrow A^*$ which maps bijectively $B$ onto $X$.

Let $S$ be a set of words. A prefix code $X\subset S$ is $S$-maximal 
if it is not properly contained in any prefix code 
$Y\subset S$. Note that if $X\subset S$ is an $S$-maximal prefix code, 
any word of $S$ is comparable for the prefix order with a word of $X$.

We denote by $X^*$ the submonoid generated by $X$.
A set $X\subset S$ is \emph{right $S$-complete} if any word of $S$
is a prefix of a word in $X^*$. Given a factorial set $S$,
a prefix code is $S$-maximal if
and only if it is right $S$-complete (Proposition 3.3.2
in~\cite{BerstelDeFelicePerrinReutenauerRindone2012}).

A \emph{parse} of a word $w$ with respect to a set $X$ is
a triple $(v,x,u)$ such that $w=vxu$ where $v$ has no suffix in $X$,
$u$ has no prefix in $X$ and $x\in X^*$.
We denote by $\delta_X(w)$ the number of parses of $w$ with respect to $X$.
Let $X$ be a prefix code. By Proposition 4.1.6 in~\cite{BerstelDeFelicePerrinReutenauerRindone2012},
for any $u\in A^*$ and $a\in A$, one has
\begin{equation}
\delta_X(ua)=\begin{cases}\delta_X(u)&\text{if $ua\in A^*X$},\\
                         \delta_X(u)+1&\text{otherwise}.
\end{cases}\label{eqparses}
\end{equation}

%%%%%%%%%%%%%%%%%%%%%%%%%%%%
\subsubsection{Maximal bifix codes}
Let $S$ be a set of words.
A bifix code $X\subset S$ is $S$-maximal if it is
not properly contained in a bifix code $Y\subset S$.
For a recurrent set $S$, a finite bifix code is $S$-maximal as a bifix code if
and only if it is an $S$-maximal prefix code 
(see~\cite{BerstelDeFelicePerrinReutenauerRindone2012}, Theorem
4.2.2).

By definition, the $S$-\emph{degree} of a bifix code
 $X$, denoted $d_X(S)$, is the maximal number
of parses of a word in $S$.  It can be finite or infinite.

For $S=A^*$, we use the term `maximal bifix code' instead of $A^*$-maximal
bifix code and `degree' instead of $A^*$-degree. This is consistent
with the terminology of~\cite{BerstelPerrinReutenauer2009}.

Let $X$ be a bifix code.
The number of parses of a word $w$ is also equal to the number
of suffixes of $w$ which have no prefix in $X$ and the
 number of prefixes of $w$ which have no suffix in $X$
(see Proposition 6.1.6 in~\cite{BerstelPerrinReutenauer2009}).

The set of \emph{internal factors} of a set of words $X$,
denoted $I(X)$,
is the set of words $w$ such that 
there exist nonempty words $u,v$ with $uwv\in X$.

Let $S$ be a set of words. A set $X\subset S$ is said to be
$S$-\emph{thin} if there is a word of $S$ which is not a factor of
$X$. If $S$ is biextendable any finite set $X\subset S$ is $S$-thin.
Indeed, any long enough word of $S$ is not a factor of $X$.
The converse is true if $S$ is uniformly recurrent. Indeed, let
$w\in S$ be a word which is not a factor of $X$. Then any long
enough word of $S$ contains $w$ as a factor, and thus is not itself a factor
of $X$.

Let $S$ be a recurrent set and let
 $X$ be an $S$-thin and $S$-maximal bifix code of $S$-degree $d$.
A word $w\in S$ is such that
$\delta_X(w)< d$ if and only if it is an internal factor of $X$,
that is,
\begin{displaymath}
I(X)=\{w\in S\mid \delta_X(w)<d\}
\end{displaymath}
(Theorem 4.2.8 in~\cite{BerstelDeFelicePerrinReutenauerRindone2012}).
Thus any word of $S$ which is not a factor of $X$ has $d$ parses.
This implies that
the $S$-degree $d$ is finite.
\begin{example}\label{exampleUniform}
Let $S$ be a recurrent set. For any integer $n\ge 1$, the set
$S\cap A^n$ is an $S$-maximal bifix code of $S$-degree $n$.
\end{example}
The \emph{kernel} of a bifix code $X$ is the set $K(X)=I(X)\cap X$.
Thus it is the set of words of $X$ which are also internal factors of
$X$.
By Theorem 4.3.11
of~\cite{BerstelDeFelicePerrinReutenauerRindone2012},
 an $S$-thin and
$S$-maximal bifix code is determined by its $S$-degree and its kernel.
Moreover, by Theorem 4.3.12 
of~\cite{BerstelDeFelicePerrinReutenauerRindone2012}, we have the
following result.
\begin{theorem}\label{theoremKernel}
Let $S$ be a recurrent set. A bifix code $Y\subset S$ is the kernel
of some $S$-thin $S$-maximal bifix code of $S$-degree $d$ if and only if
$Y$ is not $S$-maximal and $\delta_Y(y)\le d-1$ for all $y\in Y$.
\end{theorem}
\begin{example}\label{exampleKernel}
Let $S$ be the Fibonacci set. The set $Y=\{a\}$ is a bifix code which is not
$S$-maximal and $\delta_Y(a)=1$. The set $X=\{a,baab,bab\}$ is the unique
$S$-maximal bifix code of $S$-degree $2$ with kernel $\{a\}$. Indeed,
the word $bab$ is not an internal factor and has two parses,
namely $(1,bab,1)$ and $(b,a,b)$.
\end{example}
The following proposition allows one to embed an $S$-maximal
bifix code in a maximal one of the same degree.
\begin{proposition}\label{propositionCompletion}
Let $S$ be a recurrent set.
For any $S$-thin and
 $S$-maximal bifix code $X$ of $S$-degree $d$, there is a thin
maximal bifix code $X'$ of degree $d$ such that $X=X'\cap S$.
\end{proposition}
\begin{proof}
Let $K$ be the kernel of $X$ and let $d$ be the $S$-degree of $X$.
By Theorem~\ref{theoremKernel}, the set $K$ is not $S$-maximal and
$\delta_K(y)\le d-1$ for any $y\in K$.
Thus, applying again
 Theorem~\ref{theoremKernel} with $S=A^*$, 
there is a maximal bifix code $X'$ with kernel $K$ and degree $d$.
Then, by Theorem 4.2.11
of~\cite{BerstelDeFelicePerrinReutenauerRindone2012},
the set $X'\cap S$ is an $S$-maximal bifix code.

Let us show that $X\cup X'$ is prefix. Suppose that
 $x\in X$ and $x'\in  X'$ are comparable for
the prefix order. We may assume  that $x$ is a  prefix 
of $x'$ (the other case works symmetrically).
 If $x\in K$, then $x\in X'$ and thus $x=x'$. Otherwise,
$\delta_X(x)=d$.
Set $x=pa$ with $a\in A$. Then, by Equation~\eqref{eqparses}, 
$\delta_X(x)=\delta_X(p)$ and thus $\delta_X(p)=d$. But since
all the factors of $p$ which are in $X$ are in $K$, we have 
$\delta_X(p)=\delta_K(p)$. Analogously, since all factors
of $p$ which are in $X'$ are in $K$, we have
$\delta_K(p)=\delta_{X'}(p)$. Therefore $\delta_{X'}(p)=d$.
But, since $X'$ has degree $d$, $\delta_{X'}(x)\le d$.
Then, by Equation~\eqref{eqparses} again, we have $\delta_{X'}(x)=d$
and $x\in A^*X'$.
Let $z$ be the suffix of $x$ which is in $X'$. If $x\ne x'$, then
 $z=x$ or  $z\in K$ 
and in both cases $z\in X$. Since $X'$ is prefix and $X$
is suffix, this implies
 $z=x=x'$. 

Since $X$ and $X'\cap S$ are $S$-maximal prefix codes included in
$(X\cup X')\cap S$, 
this implies
that $X=X'\cap S$.
\end{proof}
\begin{example}
Let $S$ be the Fibonacci set. Let $X=\{a,baab,bab\}$ be the $S$-maximal
bifix code of $S$-degree $2$ with kernel $\{a\}$. Then
$X'=a\cup ba^*b$ is the  maximal bifix code with kernel $\{a\}$ of degree $2$
such that $X'\cap S=X$.
\end{example}
%%%%%%%%%%%%%%%%%%%%%%%%%%%%%%%%
\subsubsection{Internal transformation}
We will use the following transformation which operates on bifix
codes (see~\cite[Chapter 6]{BerstelPerrinReutenauer2009} 
for a more detailed presentation). For a set of words $X$ and a word $u$, we denote
$u^{-1}X=\{v\in A^*\mid uv\in X\}$ and $Xu^{-1}=\{v\in A^*\mid vu\in
X\}$
the \emph{residuals} of $X$ with respect to $u$ (one should not
confuse this notation with that of the inverse in the free group). 
Let $X\subset S$ be a set of words and $w\in S$ a word. Let
\begin{eqnarray}
G=Xw^{-1},\quad &D=w^{-1}X,\\
G_0=(wD)w^{-1}&D_0=w^{-1}(Gw),\label{eqG_0}\\
G_1=G\setminus G_0,&D_1=D\setminus D_0.\label{eqGD}
\end{eqnarray}
Note that $Gw\cap wD=G_0w=wD_0$. Consequently $G_0^*w=wD_0^*$.
The set
\begin{equation}
Y=(X\cup w\cup(G_1wD_0^*D_1\cap S))\setminus (Gw\cup wD)\label{eqGenInternal}
\end{equation}
is said to be obtained from $X$ by \emph{internal transformation}
with respect to $w$. When $Gw\cap wD=\emptyset$, the transformation
takes the simpler form
\begin{equation}
Y=(X\cup w\cup(GwD\cap S))\setminus (Gw\cup wD).\label{eqInternal}
\end{equation}
It is this form which is used in~\cite{BerstelDeFelicePerrinReutenauerRindone2012}
to define the internal transformation.
\begin{example}\label{exampleFiboDegre2}
Let $S$ be the Fibonacci set. Let $X=S\cap A^2$. The internal
transformation applied to
 $X$ with respect to $b$ gives $Y=\{aa,aba,b\}$. The 
internal
transformation applied to $X$ with respect to $a$ gives
$Y'=\{a,baab,bab\}$.
\end{example}
The following result is proved in~\cite{BerstelDeFelicePerrinReutenauerRindone2012} in the case $G_0=\emptyset$
(Proposition 4.4.5).
\begin{proposition}\label{propositionInternal}
Let $S$ be a uniformly recurrent set and
 let $X\subset S$ be a finite $S$-maximal
bifix code of $S$-degree $d$. Let $w\in S$ be a nonempty word
such that the sets
$G_1,D_1$ defined by Equation~\eqref{eqGD} are nonempty.
Then the set $Y$ obtained as in Equation~\eqref{eqGenInternal} is
a finite $S$-maximal bifix code with $S$-degree at most $d$.
\end{proposition}
\begin{proof}
By Proposition~\ref{propositionCompletion} there is a thin maximal bifix
code $X'$ of degree $d$ such that $X=X'\cap S$. Let $Y'$ be
the code obtained from $X'$ by internal transformation with
respect to $w$. Then
\begin{displaymath}
Y'=(X'\cup w\cup(G'_1w{D'_0}^*D'_1))\setminus (G'w\cup wD')
\end{displaymath}
with $G'=X'w^{-1}$, $D'=w^{-1}X'$, and
$G'_0=(wD')w^{-1}$, $D'_0=w^{-1}(G'w)$,
$G'_1=G'\setminus G'_0$, $D'_1=D'\setminus D'_0$. 
We have $G=G'\cap Sw^{-1}$, $D=D'\cap w^{-1}S$, and $D_i=D'_i\cap w^{-1}S$,
 $G_i=G'_i\cap Sw^{-1}$ for $i=0,1$. In particular $G_1\subset G'_1$,
$D_1\subset D'_1$. Thus $G'_1,D'_1\ne\emptyset$. This implies that
$Y'$ is a thin maximal bifix code of degree $d$ (see Proposition 6.2.8
and its complement page 242
in \cite{BerstelPerrinReutenauer2009}).

Since $w\in S$, we have $Y=Y'\cap S$. By Theorem 4.2.11
of~\cite{BerstelDeFelicePerrinReutenauerRindone2012},
$Y$ is an $S$-maximal bifix code of $S$-degree at most
$d$. Since $S$ is uniformly
recurrent,
this implies that $Y$ is finite.
\end{proof}
When $G_0=\emptyset$, the bifix code $Y$ has $S$-degree $d$ (see~\cite[Proposition 4.4.5]{BerstelDeFelicePerrinReutenauerRindone2012}). We will see in the proof of
Theorem~\ref{theoremClelia} another case where it is true. We have no example
where it is not true.
\begin{example}
Let $S$ be the Fibonacci set, as in Example~\ref{exampleFiboDegre2}. Let $X=S\cap A^2$
 and let $w=a$. Then
$Y=\{a,baab,bab\}$ is the $S$-maximal bifix code of $S$-degree $2$
already considered in  Example~\ref{exampleFiboDegre2}.
\end{example}
%%%%%%%%%%%%%%%%%%%%%%%%%%%%%%%%%%%%%%%%%%%%%
\section{Strong, weak and neutral sets}\label{sectionNeutrality}

In this section, we introduce strong, weak and neutral sets. We prove a theorem
concerning
the cardinality of an $S$-maximal bifix code in a neutral set $S$
(Theorem~\ref{theoremCardinality}). 

%%%%%%%%%%%%%%%%%%%%%%%%
\subsection{Strong, weak and neutral words}
Let $S$ be a factorial set.
For a word $w\in S$, let
\begin{displaymath}
m(w)=e(w)-\ell(w)-r(w)+1.
\end{displaymath}
We say that, with respect to
 $S$,  $w$ is \emph{strong} if $m(w)>0$,
\emph{weak} if $m(w)<0$ and \emph{neutral}  if $m(w)=0$.

A biextendable word $w$ is called \emph{ordinary} if $E(w)\subset a\times
A\cup A\times b$ for some $(a,b)\in E(w)$ (see~\cite[Chapter 4]{BertheRigo2010}).
If $S$ is biextendable, any ordinary word is neutral. Indeed, one has
$E(w)=(a\times (R(w)\setminus b))\cup ((L(w)\setminus a)\times b)\cup (a,b)$
and thus $e(w)=\ell(w)+r(w)-1$.

\begin{example}\label{exSturmianIsOrdinary}
 In a Sturmian set, any word is ordinary. Indeed, for
any bispecial word $w$, there is a unique letter $a$ such that
$aw$ is right-special and a unique letter $b$ such that $wb$
is left-special. Then $awb\in S$ and $E(w)=a\times A\cup A\times b$.
\end{example}
We say that a set of words $S$ is \emph{strong} (resp. \emph{weak},
resp. \emph{neutral}) if it is
factorial
and  every  word
$w\in S$ is strong or neutral (resp. weak or neutral, resp. neutral). 

The sequence $(p_n)_{n\ge 0}$ with $p_n=\Card(S\cap A^n)$ is called the
\emph{complexity} of $S$. Set $k=\Card(S\cap A)-1$.
\begin{proposition}\label{propComplexityNeutral}
The complexity of a strong (resp. weak, resp. neutral) set $S$
is 
at least  (resp. at most, resp. exactly) equal to $kn+1$.
\end{proposition}

Given a factorial set $S$ with complexity $p_n$, we denote
$s_n=p_{n+1}-p_n$ the first difference of the sequence $p_n$
and $b_n=s_{n+1}-s_n$ its second difference.
The following is from~\cite{Cassaigne1997}
(it is also part of Theorem 4.5.4 in~\cite[Chapter 4]{BertheRigo2010}
and also Lemma 3.3
in~\cite{BertheDeFeliceDolceLeroyPerrinReutenauerRindone2013d}).
\begin{lemma}\label{lemmaEnum}
We have 
\begin{displaymath}
b_n=\sum_{w\in A^n\cap S}m(w)\quad \text{ and } \quad s_n=\sum_{w\in A^n\cap
  S}(r(w)-1)
\end{displaymath}
for all $n\ge 0$.
\end{lemma}

Proposition~\ref{propComplexityNeutral} follows easily from
the following lemma.
\begin{lemma}\label{lemmasn}
If $S$ is strong (resp. weak, resp. neutral), then $s_n\ge k$
(resp. $s_n\le k$, resp. $s_n=k$)
for all $n\ge 0$.
\end{lemma}
\begin{proof}
Assume that $S$ is strong. Then $m(w)\ge 0$ for all $w\in S$ and
thus, by Lemma~\ref{lemmaEnum}, the sequence $(s_n)$ is nondecreasing.
Since $s_0=k$, this implies $s_n\ge k$ for all $n$. 
The proof of the other cases is similar.
\end{proof}

We now give an example of a set of complexity $2n+1$ on an alphabet
with three letters
which is not neutral.
\begin{example}\label{exampleChacon}
Let $A=\{a,b,c\}$.
The \emph{Chacon word} on three letters
is the fixpoint $x=f^\omega(a)$ of the morphism $f$ from
$A^*$ into itself defined by $f(a)=aabc$, $f(b)=bc$ and $f(c)=abc$.
Thus $x=aabcaabcbcabc\cdots$. The \emph{Chacon set} is the set $S$ of
factors of $x$. It is of complexity $2n+1$ (see~\cite[Section 5.5.2]{PytheasFogg2002}).

It contains strong, neutral and weak words. Indeed,
$S\cap A^2=\{aa,ab,bc,ca,cb\}$ and thus $m(\varepsilon)=0$
showing that the empty word is neutral. Next
$E(abc)=\{(a,a), (c,a), (a,b),(c,b)\}$ shows that $m(abc)=1$
and thus $abc$ is strong. Finally, $E(bca)=\{(a,a),(c,b)\}$
and thus $m(bca)=-1$ showing that $bca$ is weak.
\end{example}

%%%%%%%%%%%%%%%%%%%%%%%%
\subsection{The Cardinality Theorem}
The following result,  referred to as the Cardinality Theorem,
is a generalization of
a result proved in \cite{BerstelDeFelicePerrinReutenauerRindone2012} in the less
general case of a Sturmian set. Since the set $S\cap A^n$ is an $S$-maximal
bifix
code of $S$-degree $n$ (see Example~\ref{exampleUniform}), it is also a
generalization
of Proposition~\ref{propComplexityNeutral}.

\begin{theorem}\label{theoremCardinality}
Let $S$ be a recurrent set containing the alphabet
$A$ and let $X\subset S$ be a finite $S$-maximal
 bifix code. Set $k=\Card(A)-1$ and $d=d_X(S)$.
  If $S$ is strong (resp. weak), then
 $\Card(X)-1\ge dk$ (resp. $\Card(X)-1\le dk$). If $S$ is neutral, 
then $\Card(X)-1=dk$. 
\end{theorem}

Note that, for a recurrent neutral set $S$,
 a bifix code $X\subset S$ may be infinite since this may happen
for a Sturmian set $S$
(see~\cite[Example 5.1.4]{BerstelDeFelicePerrinReutenauerRindone2012}).

We consider rooted trees with the usual notions of root, node, child and parent.
The following lemma is an application of a well-known lemma on trees
relating the number of its leaves to the sum of the degrees
of its internal nodes.
\begin{lemma}\label{lemmaArity}
Let $S$ be a prefix-closed set.
Let $X$ be a finite $S$-maximal  prefix code and let $P$ be the set of its proper prefixes.
 Then
$\Card(X)=1+\sum_{p\in P}(r(p)-1)$.
\end{lemma}

We order the nodes of a tree from the parent to the child and thus
we have $m\le n$ if $m$ is a descendant of $n$. We denote $m<n$
if $m\le n$ with $m\ne n$.

\begin{lemma}\label{lemmaCombinat}
Let $T$ be a finite tree with root $r$ on a set $N$ of nodes, let $d\ge 1$, and let $\pi,\alpha$
  be  functions
assigning to each node an integer such that
\begin{enumerate}
\item[\rm(i)]for each internal node $n$,  $\pi(n)\le \sum\pi(m)$ where the sum
  runs over the children of $n$,
\item[\rm(ii)]
 for each leaf $m$ of $T$, one has $\sum_{m\le n}\alpha(n)=d$.
\end{enumerate}
Then $\sum_{n\in N}\alpha(n)\pi(n)\ge d\pi(r)$.
\end{lemma}
\begin{proof}
We use an induction on the number of nodes of $T$. If $T$ is reduced
to its root, then $d=\alpha(r)$ implies $\alpha(r)\pi(r)=d\pi(r)$ and the
result is true. 
Assume that
it holds for trees with less nodes than $T$. Since $T$ is finite
and not reduced to its root,
there is an internal node such that all its children are leaves of $T$.
Let $m$ be such a node.
Since  $\sum_{x\le n}\alpha(n)=\alpha(x)+\sum_{m\le n}\alpha(n)$ 
has  value $d$ for each child $x$ of $m$,
the value $v=\alpha(x)$ is the same for all children of $m$.
 Let $T'$ be the tree  obtained from $T$ by
deleting all children of $m$. Let $N'$ be the set of nodes of $T'$.
Let $\pi'$ be the restriction
of $\pi$ to $N'$ and let $\alpha'$ be defined by
\begin{displaymath}
\alpha'(n)=\begin{cases}\alpha(n)&\text{if $n\ne m$},\\
                        \alpha(m)+v&\text{otherwise.}
\end{cases}
\end{displaymath}
It is easy to verify that $T',\pi'$ and $\alpha'$ satisfy
the same hypotheses as $T,\pi$ and $\alpha$. Then 
\begin{eqnarray*}
\sum_{n\in N}\alpha(n)\pi(n)&=&
\sum_{n\in N'\setminus
  m}\alpha(n)\pi(n)+\alpha(m)\pi(m)+\sum_{x< m}v\pi(x)\\
&=&\sum_{n\in N'\setminus
  m}\alpha'(n)\pi'(n)+\alpha(m)\pi(m)+v\sum_{x< m}\pi(x)\\
&\ge&\sum_{n\in N'\setminus
  m}\alpha'(n)\pi'(n)+(\alpha(m)+v)\pi(m)\\
&=&\sum_{n\in N'\setminus
  m}\alpha'(n)\pi'(n)+\alpha'(m)\pi'(m)=\sum_{n\in N'}\alpha'(n)\pi'(n),
\end{eqnarray*}
whence the result by the induction hypothesis.
\end{proof}
A symmetric statement holds replacing the inequality
in condition (i) by $\pi(n)\ge\sum\pi(m)$ and the conclusion
by $\sum_{n\in N}\alpha(n)\pi(n)\le d\pi(r)$.\\

\begin{proofof}{of Theorem~\ref{theoremCardinality}}
Assume first that $S$ is strong.
Let  $N$  be larger than the lengths of the words of $X$.

Let $U$ be the set  of words of $S$ of length at most $N$. By considering each word
$w$ as the father of $aw$ for $a\in A$, the set $U$ can be considered as a tree
$T$ with root the empty word $\varepsilon$. The leaves of $T$ are the
elements of $S$ of length $N$.

For $w\in U$, set
$\pi(w)=r(w)-1$ and let 
\begin{displaymath}
\alpha(n)=\begin{cases}1&\text{ if $n$ is a proper prefix of $X$}\\ 
       0&\text{ otherwise.}
\end{cases}
\end{displaymath}
 Let us verify that
the conditions of Lemma~\ref{lemmaCombinat} are satisfied.
Let $u$ be in $U$ with $|u|<N$.
Then,
since $u$ is strong or neutral,
$\sum_{a\in L(u)}(r(au)-1)=e(u)-\ell(u)\ge r(u)-1$. This implies that
$\sum_{au\in S}\pi(au)\ge \pi(u)$ showing that condition (i) is satisfied.

Let $w$ be a leaf of $T$, that is, a word of $S$ of length $N$.
Since $N$ is larger than the maximal length of the words of $X$,
the word $w$ is not an internal factor of $X$
and thus it has $d$ parses with respect to $X$.
It implies that it has $d$ suffixes which are proper prefixes of $X$
(since $X$ is right $S$-complete, this is the same as to have no
prefix
in $X$).
Thus $\sum_{w\le u}\alpha(u)=d$. Thus condition (ii) is also satisfied.

By Lemma~\ref{lemmaCombinat}, we have
$\sum_{n\in U}\alpha(n)\pi(n)\ge d\pi(\varepsilon)$. Let $P$ be the set of proper
prefixes
of $X$. By definition of $\alpha$, we have
$\sum_{n\in U}\alpha(n)\pi(n)=\sum_{p\in P}\pi(p)$ and thus by
definition
of $\pi$, $d\pi(\varepsilon)=dk\le\sum_{p\in P}(r(p)-1)$. Since $S$
is recurrent, $X$ is an $S$-maximal prefix code. Thus, by Lemma~\ref{lemmaArity},
we have $\Card(X)=1+\sum_{p\in P}(r(p)-1)$ and thus we obtain 
$\Card(X)\ge 1+dk$ which is the desired conclusion.

The proof that $\Card(X)-1\le dk$ if $S$ is weak is
symmetric, using the symmetric version of Lemma~\ref{lemmaCombinat}.
The case where $S$ is neutral follows then directly.
\end{proofof}
We illustrate Theorem~\ref{theoremCardinality} in the following example.
\begin{example}\label{exampleCardinalityTheorem}

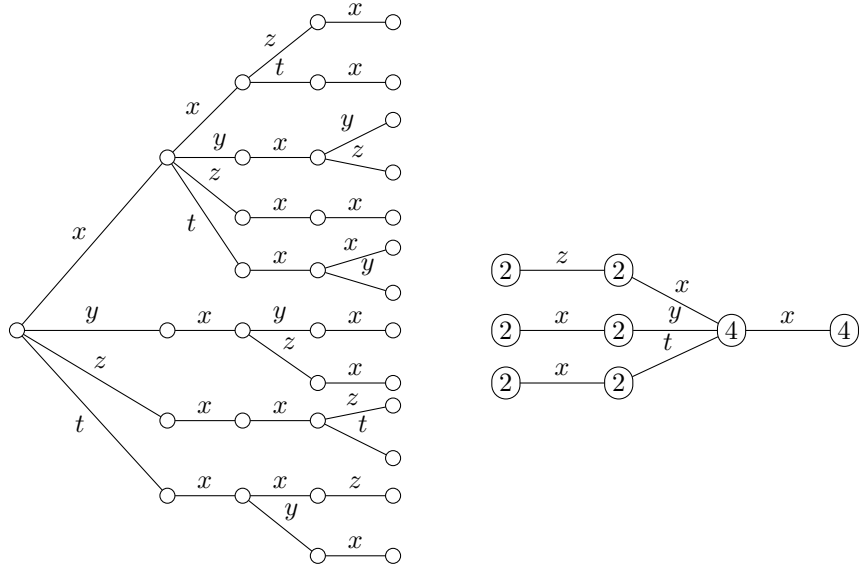
\begin{figure}[hbt]
\gasset{Nadjust=wh,AHnb=0}\centering
\begin{picture}(100,70)(0,-7)
\node(1)(0,22){}\node(x)(20,45){}\node(y)(20,22){}\node(z)(20,10){}\node(t)(20,0){}
\node(xx)(30,55){}\node(xy)(30,45){}\node(xz)(30,37){}\node(xt)(30,30){}
\node(yx)(30,22){}\node(zx)(30,10){}\node(tx)(30,0){}
\node(xxz)(40,63){}\node(xxt)(40,55){}
\node(xyx)(40,45){}\node(xzx)(40,37){}\node(xtx)(40,30){}
\node(yxy)(40,22){}\node(yxz)(40,15){}
\node(zxx)(40,10){}\node(txx)(40,0){}\node(txy)(40,-8){}
\node(xxzx)(50,63){}\node(xxtx)(50,55){}
\node(xyxy)(50,50){}\node(xyxz)(50,43){}
\node(xzxx)(50,37){}\node(xtxx)(50,33){}\node(xtxy)(50,27){}
\node(yxyx)(50,22){}\node(yxzx)(50,15){}
\node(zxxz)(50,12){}\node(zxxt)(50,5){}
\node(txxz)(50,0){}\node(txyx)(50,-8){}

\drawedge(1,x){$x$}\drawedge(1,y){$y$}\drawedge(1,z){$z$}
\drawedge[ELside=r](1,t){$t$}
\drawedge(x,xx){$x$}\drawedge[ELpos=70](x,xy){$y$}\drawedge(x,xz){$z$}\drawedge[ELside=r](x,xt){$t$}
\drawedge(y,yx){$x$}\drawedge(z,zx){$x$}\drawedge(t,tx){$x$}
\drawedge(xx,xxz){$z$}\drawedge(xx,xxt){$t$}
\drawedge(xy,xyx){$x$}\drawedge(xz,xzx){$x$}\drawedge(xt,xtx){$x$}
\drawedge(yx,yxy){$y$}\drawedge(yx,yxz){$z$}
\drawedge(zx,zxx){$x$}\drawedge(tx,txx){$x$}\drawedge(tx,txy){$y$}
\drawedge(xxz,xxzx){$x$}\drawedge(xxt,xxtx){$x$}
\drawedge(xyx,xyxy){$y$}\drawedge(xyx,xyxz){$z$}
\drawedge(xzx,xzxx){$x$}\drawedge(xtx,xtxx){$x$}\drawedge[ELpos=60](xtx,xtxy){$y$}
\drawedge(yxy,yxyx){$x$}\drawedge(yxz,yxzx){$x$}
\drawedge(zxx,zxxz){$z$}\drawedge(zxx,zxxt){$t$}
\drawedge(txx,txxz){$z$}\drawedge(txy,txyx){$x$}

\node(xtx)(65,15){$2$}\node(xyx)(65,22){$2$}\node(zxx)(65,30){$2$}
\node(tx)(80,15){$2$}\node(yx)(80,22){$2$}\node(xx)(80,30){$2$}
\node(x)(95,22){$4$}\node(1)(110,22){$4$}

\drawedge(xtx,tx){$x$}\drawedge(xyx,yx){$x$}\drawedge(zxx,xx){$z$}
\drawedge(tx,x){$t$}\drawedge(yx,x){$y$}\drawedge(xx,x){$x$}
\drawedge(x,1){$x$}

\end{picture}
\caption{The words of length at most $4$ of a neutral set
$G$  and the tree of
  right-special
words.}\label{figxyzt}
\end{figure}
Consider the set $G$ of words on the alphabet $B=\{x,y,z,t\}$ obtained
as follows.
Let $S$ be the Fibonacci set and let $X\subset S$ be the $S$-maximal
bifix code of $S$-degree $3$ defined by $X=\{a,baabaab,baabab,babaab\}$. We
consider 
the morphism $f:B^*\rightarrow A^*$ defined by $f(x)=a$,
$f(y)=baabaab$, $f(z)=baabab$, $f(t)=babaab$. We set $G=f^{-1}(S)$.

The words of $G$ of length at most $4$ are
represented in Figure~\ref{figxyzt} on the left.

Since $S$ is Sturmian, it is a uniformly recurrent
tree set (see the definition in 
Section~\ref{sectionTreePlanarTree}). By the main result of~\cite{BertheDeFeliceDolceLeroyPerrinReutenauerRindone2013c}, the family of uniformly recurrent tree
sets is closed under maximal bifix decoding.
 Thus  $G$ is a uniformly recurrent tree set.

The tree of right-special words is represented on the right
in Figure~\ref{figxyzt} with the value of $r$ indicated at each
node. The bifix codes 
\begin{displaymath}
Y=\{xx,xyx,xz,xt,y,zx,tx\},\quad
Z=\{x,yxy,yxz,zxxz,zxxt,txxz,txy\}
\end{displaymath}
are $G$-maximal and have both $G$-degree $2$. 
In agreement with Theorem~\ref{theoremCardinality}, we have
$\Card(Y)=\Card(Z)=1+2(\Card(B)-1)=7$. The codes $Y$ and $Z$
are represented in Figure~\ref{figMaxBifix}.
\begin{figure}[hbt]
\gasset{Nadjust=wh,AHnb=0}\centering
\begin{picture}(100,35)
\node[fillgray=0](1)(0,15){}\node[ExtNL=y,Nframe=n](1bis)(0,15){$3$}
\node[fillgray=0](x)(15,25){}\node[ExtNL=y,Nframe=n](xbis)(15,25){$3$}
\node[Nmr=0](y)(15,16){}\node(z)(15,8){}\node(t)(15,0){}
\node[Nmr=0](xx)(30,35){}\node(xy)(30,28){}\node[Nmr=0](xz)(30,22){}
\node[Nmr=0](xt)(30,15){}\node[Nmr=0](zx)(30,8){}\node[Nmr=0](tx)(30,0){}
\node[Nmr=0](xyx)(45,28){}

\drawedge(1,x){$x$}\drawedge(1,y){$y$}\drawedge(1,z){$z$}\drawedge[ELside=r](1,t){$t$}
\drawedge(x,xx){$x$}\drawedge[ELpos=70](x,xy){$y$}\drawedge[ELpos=60](x,xz){$z$}\drawedge(x,xt){$t$}
\drawedge(z,zx){$x$}\drawedge(t,tx){$x$}\drawedge(xy,xyx){$x$}

\node[fillgray=0](Y1)(50,20){}\node[ExtNL=y,Nframe=n](Y1bis)(50,20){$3$}
\node[Nmr=0](Yx)(65,30){}\node(Yy)(65,19){}\node(Yz)(65,13){}\node(Yt)(65,5){}
\node[fillgray=0](Yyx)(80,21){}\node[ExtNL=y,Nframe=n](Yyxbis)(80,21){$1$}
\node(Yzx)(80,13){}
\node[fillgray=0](Ytx)(80,5){}\node[ExtNL=y,Nframe=n](Ytxbis)(80,5){$1$}
\node[Nmr=0](Ytxy)(95,-2){}
\node[Nmr=0](Yyxy)(95,30){}\node[Nmr=0](Yyxz)(95,21){}
\node[fillgray=0](Yzxx)(95,13){}\node[ExtNL=y,Nframe=n](Yzxxbis)(95,13){$1$}
\node(Ytxx)(95,5){}
\node[Nmr=0](Yzxxz)(110,18){}\node[Nmr=0](Yzxxt)(110,13){}\node[Nmr=0](Ytxxz)(110,5){}

\drawedge(Y1,Yx){$x$}\drawedge[ELpos=60](Y1,Yy){$y$}\drawedge[ELpos=60](Y1,Yz){$z$}\drawedge[ELside=r](Y1,Yt){$t$}
\drawedge(Yy,Yyx){$x$}\drawedge(Yz,Yzx){$x$}\drawedge(Yt,Ytx){$x$}
\drawedge(Yyx,Yyxy){$y$}\drawedge(Yyx,Yyxz){$z$}\drawedge(Yzx,Yzxx){$x$}
\drawedge(Ytx,Ytxx){$x$}\drawedge(Ytx,Ytxy){$y$}
\drawedge(Yzxx,Yzxxz){$z$}\drawedge[ELpos=70](Yzxx,Yzxxt){$t$}
\drawedge(Ytxx,Ytxxz){$z$}
\end{picture}
\caption{Two  $G$-maximal bifix codes of $G$-degree
  $2$.}\label{figMaxBifix}
\end{figure}
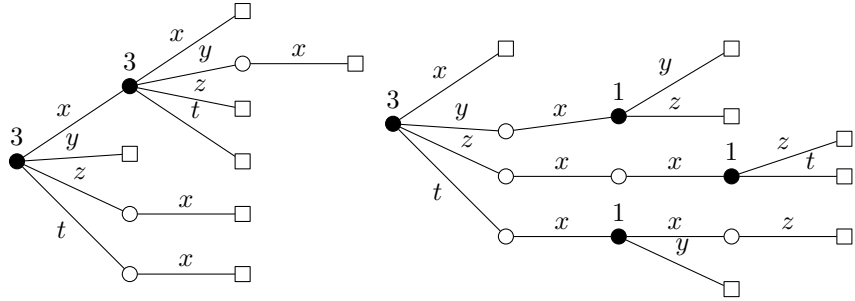
The right-special proper prefixes $p$ of $Y$ and $Z$ are
indicated in black in Figure~\ref{figMaxBifix} with the value of $r(p)-1$
indicated
for each one. In agreement with Lemma~\ref{lemmaArity},
the sum of the values of $r(p)-1$ is $6$ in both cases.

\end{example}

The following example illustrates the necessity of the hypotheses in
Theorem~\ref{theoremCardinality}.

\begin{example}\label{exampleBifixeChacon}
Consider again the Chacon set $S$ of Example~\ref{exampleChacon}.
Let $X=S\cap A^4$ and let $Y,Z$ be the $S$-maximal bifix codes of
$S$-degree $4$  represented in
Figure~\ref{figureCodeChacon}. The first one is
obtained from $X$ by internal transformation 
with respect to $abc$ . The second one with respect to $bca$.
\begin{figure}[hbt]
\gasset{Nadjust=wh,AHnb=0}
\centering
\begin{picture}(120,50)
\put(0,0){
\begin{picture}(60,50)(0,-3)
\node(1)(0,20){}
\node(a)(10,30){}\node(b)(10,20){}\node(c)(10,10){}
\node(aa)(20,40){}\node(ab)(20,30){}
\node(bc)(20,20){}
\node(ca)(20,10){}\node(cb)(20,0){}
\node(aab)(30,40){}\node[Nmr=0](abc)(30,30){}
\node(bca)(30,25){}\node(bcb)(30,17){}
\node(caa)(30,12){}\node(cab)(30,5){}
\node(cbc)(30,0){}
\node(aabc)(40,40){}
\node[Nmr=0](bcaa)(40,29){}\node[Nmr=0](bcab)(40,21){}
\node[Nmr=0](bcbc)(40,17){}
\node[Nmr=0](caab)(40,12){}
\node(cabc)(40,5){}
\node[Nmr=0](cbca)(40,0){}
\node[Nmr=0](aabca)(50,44){}\node[Nmr=0](aabcb)(50,36){}
\node[Nmr=0](cabca)(50,9){}\node[Nmr=0](cabcb)(50,1){}

\drawedge(1,a){$a$}\drawedge(1,b){$b$}\drawedge(1,c){$c$}
\drawedge(a,aa){$a$}\drawedge(a,ab){$b$}
\drawedge(b,bc){$c$}
\drawedge(c,ca){$a$}\drawedge(c,cb){$b$}
\drawedge(aa,aab){$b$}
\drawedge(ab,abc){$c$}
\drawedge(bc,bca){$a$}\drawedge(bc,bcb){$b$}
\drawedge(ca,caa){$a$}\drawedge(ca,cab){$b$}
\drawedge(cb,cbc){$c$}
\drawedge(aab,aabc){$c$}
\drawedge(bca,bcaa){$a$}\drawedge(bca,bcab){$b$}
\drawedge(bcb,bcbc){$c$}
\drawedge(caa,caab){$b$}
\drawedge(cab,cabc){$c$}
\drawedge(cbc,cbca){$a$}
\drawedge(aabc,aabca){$a$}\drawedge(aabc,aabcb){$b$}
\drawedge(cabc,cabca){$a$}\drawedge(cabc,cabcb){$b$}
\end{picture}
}
\put(60,0){
\begin{picture}(60,50)
\node(1)(0,20){}
\node(a)(10,30){}\node(b)(10,20){}\node(c)(10,10){}
\node(aa)(20,40){}\node(ab)(20,30){}
\node(bc)(20,20){}
\node(ca)(20,10){}\node(cb)(20,0){}
\node(aab)(30,40){}\node(abc)(30,30){}
\node[Nmr=0](bca)(30,25){}\node(bcb)(30,17){}
\node(caa)(30,12){}\node(cab)(30,5){}
\node(cbc)(30,0){}
\node[Nmr=0](aabc)(40,40){}
\node(abca)(40,35){}\node[Nmr=0](abcb)(40,25){}
\node[Nmr=0](abcaa)(50,35){}
\node[Nmr=0](bcbc)(40,17){}
\node[Nmr=0](caab)(40,12){}
\node[Nmr=0](cabc)(40,5){}
\node(cbca)(40,0){}
\node[Nmr=0](cbcab)(50,0){}

\drawedge(1,a){$a$}\drawedge(1,b){$b$}\drawedge(1,c){$c$}
\drawedge(a,aa){$a$}\drawedge(a,ab){$b$}
\drawedge(b,bc){$c$}
\drawedge(c,ca){$a$}\drawedge(c,cb){$b$}
\drawedge(aa,aab){$b$}
\drawedge(ab,abc){$c$}
\drawedge(bc,bca){$a$}\drawedge(bc,bcb){$b$}
\drawedge(ca,caa){$a$}\drawedge(ca,cab){$b$}
\drawedge(cb,cbc){$c$}
\drawedge(aab,aabc){$c$}
\drawedge(abc,abca){$a$}\drawedge(abc,abcb){$b$}
\drawedge(abca,abcaa){$a$}
\drawedge(bcb,bcbc){$c$}
\drawedge(caa,caab){$b$}
\drawedge(cab,cabc){$c$}
\drawedge(cbc,cbca){$a$}

\drawedge(cbca,cbcab){$b$}
\end{picture}
}
\end{picture}
\caption{Two $S$-maximal bifix codes of $S$-degree $4$.}\label{figureCodeChacon}
\end{figure}
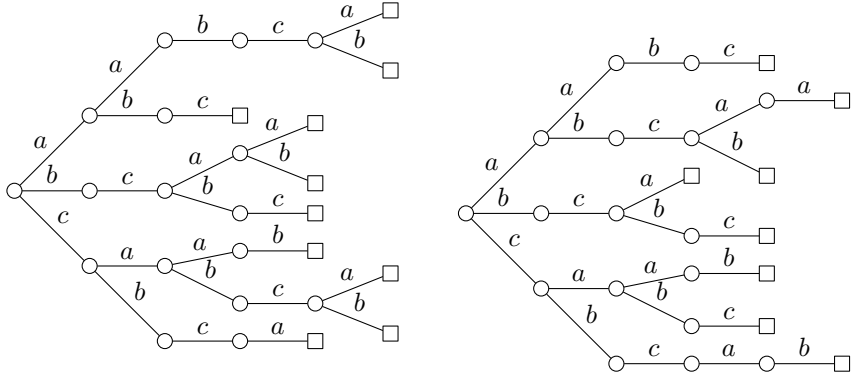
We have $\Card(Y)=10$ and $\Card(Z)=8$ showing that $\Card(Y)-1>8$ and
$\Card(Z)-1<8$, illustrating the fact that $S$ is neither strong nor
weak.
\end{example}

The following example shows that the class of sets of factor complexity
$kn+1$ is not closed by maximal bifix decoding.
\begin{example}\label{exampleBifixeChacon2}
Let $S$ be the Chacon set and let $f:B^*\rightarrow A^*$
 be a coding morphism for the $S$-maximal bifix code $Z$ of $S$-degree $4$
with $8$ 
elements of Example~\ref{exampleBifixeChacon}.
One may verify that $\Card(B^2\cap f^{-1}(S))=\Card(Z^2\cap S)=17$. 
This shows that the set $f^{-1}(S)$ does not have factor complexity
$7n+1$.
\end{example}
%%%%%%%%%%%%%%%%%%%%%%%%%%%%
\subsection{A converse of the Cardinality Theorem}\label{sectionConverseCardinality}
We end this section with a statement proving a converse of the
Cardinality Theorem.
\begin{theorem}\label{theoremClelia}
Let $S$ be a uniformly recurrent set containing the alphabet $A$. If  any finite $S$-maximal
bifix code of $S$-degree $d$ has $d(\Card(A)-1)+1$ elements, then
$S$ is neutral.
\end{theorem}
\begin{proof}
We may assume that $A$ has more than one element.
We argue by contradiction. Let $w\in S$ be a word which is not
neutral. We cannot have $w=\varepsilon$ since otherwise the
$S$-maximal bifix code $X=S\cap A^2$ has not the good cardinality.

Set $n=|w|$ and $X=S\cap A^{n+1}$. The set $X$
is an $S$-maximal bifix code of $S$-degree $n+1$.
Let $Y$ be the code obtained by internal
transformation from $X$ with respect to $w$ and defined by 
Equation~\eqref{eqGenInternal}. Note that $G=L(w)$
and $D=R(w)$.

We distinguish two cases.

\paragraph{Case 1.} 
Assume that $Gw\cap wD=\emptyset$. 

 The code $Y$ is
 defined by Equation~\eqref{eqInternal} and we have $\Card(GwD\cap S)=e(w)$.
Since $D_0=G_0=\emptyset$, the hypotheses of Proposition~\ref{propositionInternal}
are satisfied and $Y$ has $S$-degree $n+1$ (by Proposition 4.4.5 in~\cite{BerstelDeFelicePerrinReutenauerRindone2012}).
 This implies
$\Card(X)=\Card(Y)$. On the other hand
\begin{displaymath}
\Card(Y)=\Card(X)+1+e(w)-\ell(w)-r(w)=\Card(X)+m(w).
\end{displaymath}
Since $w$ is not neutral, we have $m(w)\ne 0$  and thus we obtain a contradiction.
\paragraph{Case 2.}
Assume next that $Gw\cap wD\ne\emptyset$. Then $w=a^n$ with $n>0$ for some
letter $a$ and the sets $G_0,D_0$ defined by Equation~\eqref{eqG_0}
are $G_0=D_0=\{a\}$. Moreover $a^{n+1}\in X$.

Since $w$ is not neutral, it is bispecial.
Thus the sets $G_1,D_1$ are nonempty and the hypotheses of
Proposition~\ref{propositionInternal} are satisfied. 
Since $S$ is uniformly
recurrent and since $S\ne a^*$,
the set $a^*\cap S$ is finite. Set $a^* \cap S=\{1,a,\ldots,a^m\}$.
Thus $m\ge n+1$.

Let $b\ne a$ be a letter such that $a^mb\in S$. Then, $\delta_Y(a^m)=n$
since $a^m$ has $n$ suffixes which are proper prefixes of $Y$.
Moreover, $a^mb$ has no suffix in $Y$. Indeed, if $a^tb\in Y$,
we cannot have $t\ge n$ since $a^n\in Y$. And since all words in $Y$
except $a^n$
have length greater than $n$, $t<n$ is also impossible.
Thus by Equation~\eqref{eqparses},
we have
$\delta_Y(a^mb)=\delta_Y(a^m)+1$ and thus $\delta_Y(a^mb)=n+1$. This shows
that the $S$-degree of $Y$ is $n+1$ and thus that $\Card(Y)=\Card(X)$
as in Case 1.

 We may assume that $n$ is chosen
maximal such that $a^n$ is not neutral. This is always possible
if $a^m$ is neutral. Otherwise, Case 1 applies to $X=S\cap A^{m+1}$ and $w=a^m$.

For $n\le i\le m-2$ (there may be
no such integer $i$ if $n=m-1$), since $a^{i+1}$ is neutral, we have
\begin{displaymath}
\Card(G_1a^iD_1\cap S)=e(a^i)-\ell(a^{i+1})-r(a^{i+1})+1=e(a^i)-e(a^{i+1}).
\end{displaymath}
Moreover, $\Card(G_1a^{m-1}D_1\cap S)=e(a^{m-1})-r(a^m)-\ell(a^m)=e(a^{m-1})-e(a^m)-1$ and
$\Card(G_1a^mD_1\cap S)=e(a^m)$. Thus
\begin{eqnarray*}
\Card(G_1a^na^
*D_1\cap S)&=&\sum_{i=n}^{m-2}(e(a^i)-e(a^{i+1}))+e(a^{m-1})-e(a^m)-1+e(a^m)\\
&=&e(a^n)-1.
\end{eqnarray*}
Hence $\Card(Y)-\Card(X)$ evaluates as
\begin{eqnarray*}
&&1+\Card(G_1a^na^*D_1\cap S)-\Card(Ga^n)-\Card(a^nD)+1\\
&=&1+e(a^n)-1-\ell(a^n)-r(a^n)+1\\
&=& m(a^n)
\end{eqnarray*}
(the last $+1$ on the first line comes from the word $a^{n+1}$
counted twice in $\Card(Gw)+\Card(wD)$). Since $m(a^n)\ne 0$,
this contradicts the fact that $X$ and $Y$ have the same number of elements.
\end{proof}

%%%%%%%%%%%%%%%%
\section{Tree sets}\label{sectionTreePlanarTree}
We  introduce in this section the notions of acyclic and
tree sets.
We state and prove the main result of this
paper (Theorem~\ref{newTheoremBasis}). The proof uses results
from \cite{BertheDeFeliceDolceLeroyPerrinReutenauerRindone2013d}.
%%%%%%%%%%%%%%%%%%%%
\subsection{Acyclic and tree sets}
Let $S$ be a set of words.
For  $w\in S$,  the \emph{extension graph} $G(w)$
of $w$ is the following undirected bipartite graph. Its set of vertices
 is the disjoint union of two copies of the sets
$L(w)$ and $R(w)$. Next, its edges are the pairs 
$(a,b)\in E(w)$. By definition of $E(w)$,
an edge  goes from $a\in L(w)$
to $b\in R(w)$ if and only if $awb\in S$.

Recall that an undirected graph is a tree if it is connected and acyclic.

Let $S$ be a  biextendable set. We say that $S$ is \emph{acyclic} 
  if for every  word $w\in S$, the graph
 $G(w)$ is
acyclic. We say that $S$ is a \emph{tree set}
if $G(w)$ is a tree for all $w\in S$.

Clearly an acyclic set is weak and a tree set is neutral.

Note  that a biextendable set $S$ is a tree set
if and only if the graph $G(w)$ is a tree for every bispecial 
non-ordinary word
$w$.
Indeed, if $w$ is not bispecial or if it is ordinary, then  $G(w)$
is  always a tree. 

\begin{proposition}\label{propositionSturmIsTree}
A Sturmian set $S$ is a tree set.
\end{proposition}
Indeed, $S$ is biextendable and every bispecial word is ordinary
(see Example~\ref{exSturmianIsOrdinary}).

The following example shows that there are neutral sets
which are not tree sets.
\begin{example}\label{exampleNeutralNotTree}
Let $A=\{a,b,c\}$ and let $S$ be the set of factors of
$a^*\{bc,bcbc\} a^*$. The set $S$ is
biextendable. One has $S\cap A^2=\{aa,ab,bc,cb,ca\}$.
It is neutral. Indeed the empty word is neutral
since $e(\varepsilon)=\Card(S\cap A^2)=5=\ell(\varepsilon)+r(\varepsilon)-1$. 
Next, the only nonempty bispecial words
are $bc$ and 
$a^n$ for $n\ge 1$. They are neutral since $e(bc)=3=\ell(bc)+r(bc)-1$
and $e(a^n)=3=\ell(a^n)+r(a^n)-1$. However, $S$
is not  acyclic 
 since the graph $G(\varepsilon)$ contains
a cycle (and has two connected components, see Figure~\ref{figureChacon}).
\begin{figure}[hbt]
\centering\gasset{Nadjust=wh,AHnb=0}
\begin{picture}(20,10)
\node(a1)(0,10){$a$}\node(a2)(20,10){$a$}
\node(b1)(0,0){$b$}\node(b2)(20,5){$b$}
\node(c1)(0,5){$c$}\node(c2)(20,0){$c$}

\drawedge(a1,a2){}\drawedge(a1,b2){}
\drawedge(b1,c2){}
\drawedge(c1,a2){}\drawedge(c1,b2){}
\end{picture}
\caption{The graph $G(\varepsilon)$ for the set $S$.}\label{figureChacon}
\end{figure}
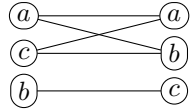
\end{example}
In the last example, the set is not recurrent. We present now
an example, due to Julien Cassaigne~\cite{Cassaigne2013} 
 of a uniformly recurrent set which is neutral
but is not a tree set (it is actually not even acyclic).

\begin{example}\label{exampleJulien}
Let $A=\{a,b,c,d\}$ and let $\sigma$ be the morphism from $A^*$
into itself defined by
\begin{displaymath}
\sigma(a)=ab,\ \sigma(b)=cda,\ \sigma(c)=cd,\ \sigma(d)=abc.
\end{displaymath}

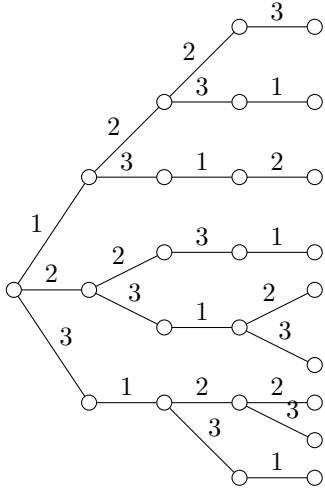
\begin{figure}[hbt]
\centering
\gasset{Nadjust=wh,AHnb=0}
\begin{picture}(50,60)

\begin{picture}(40,60)
\node(0)(0,25){}
\node(1)(10,40){}\node(2)(10,25){}\node(3)(10,10){}
\node(12)(20,50){}\node(13)(20,40){}
\node(22)(20,30){}\node(23)(20,20){}
\node(31)(20,10){}
\node(122)(30,60){}\node(123)(30,50){}
\node(131)(30,40){}
\node(223)(30,30){}
\node(231)(30,20){}
\node(312)(30,10){}\node(313)(30,0){}
\node(1223)(40,60){}\node(1231)(40,50){}
\node(1312)(40,40){}
\node(2231)(40,30){}
\node(2312)(40,25){}\node(2313)(40,15){}
\node(3122)(40,10){}\node(3123)(40,5){}
\node(3131)(40,0){}

\drawedge(0,1){$1$}\drawedge(0,2){$2$}\drawedge(0,3){$3$}
\drawedge(1,12){$2$}\drawedge(1,13){$3$}
\drawedge(2,22){$2$}\drawedge(2,23){$3$}
\drawedge(3,31){$1$}
\drawedge(12,122){$2$}\drawedge(12,123){$3$}
\drawedge(13,131){$1$}
\drawedge(22,223){$3$}
\drawedge(23,231){$1$}
\drawedge(31,312){$2$}\drawedge(31,313){$3$}
\drawedge(122,1223){$3$}\drawedge(123,1231){$1$}
\drawedge(131,1312){$2$}
\drawedge(223,2231){$1$}
\drawedge(231,2312){$2$}\drawedge(231,2313){$3$}
\drawedge(312,3122){$2$}\drawedge[ELpos=60](312,3123){$3$}
\drawedge(313,3131){$1$}
\end{picture}

\end{picture}
\caption{The words of length at most $4$ of the set $S$.}\label{figureCassaigne}
\end{figure}

Let $B=\{1,2,3\}$ and let $\tau:A^*\rightarrow B^*$ be defined by
\begin{displaymath}
\tau(a)=12,\quad \tau(b)=2,\quad \tau(c)=3,\quad \tau(d)=13.
\end{displaymath}
Let $S$ be the set of factors of the infinite word $\tau(\sigma^\omega(a))$
(see Figure~\ref{figureCassaigne}).

It is shown in~\cite[Example 4.5]{BertheDeFeliceDolceLeroyPerrinReutenauerRindone2013d} 
 that $S$ is a uniformly recurrent neutral set.
 It is not a tree set since $G(\varepsilon)$ is neither acyclic nor
connected.
\end{example}
%%%%%%%%%%%%%%%%%%%%%

%%%%%%%%%%%%%%%%%%%%%%%%%%%%%%%%%%%%%%%%%%%%%%%%%%%%%%%%%%%%%

%%%%%%%%%%%%%%%%%%%%%%%%%%%%%%%%%%%%%%%%%%%%%
%\section{Bifix codes in tree sets}\label{sectionBasisTheorem}

\subsection{Finite index basis property}\label{sectionBasisTheorem}

Let $S$ be a recurrent set containing the alphabet $A$.
We say that $S$ has the \emph{finite index basis property} if the
following holds:
a  finite bifix code $X\subset S$ is an  $S$-maximal bifix code of $S$-degree
$d$ if and only if it is a basis of a subgroup of index $d$ of the
free group on $A$. 

We will prove the following result, referred to as the Finite Index
Basis Theorem. 

\begin{theorem}\label{newTheoremBasis}
Any uniformly recurrent tree set $S$  containing the alphabet $A$
has the finite index basis property.
\end{theorem}

Note that the Cardinality Theorem (Theorem~\ref{theoremCardinality})
holds
for a set $S$ satisfying the finite index basis property.
Indeed, by Schreier's formula a basis of a subgroup
of index $d$ of a free group on $s$ generators has $(s-1)d+1$
elements (actually we use Theorem \ref{theoremCardinality} in the
proof of Theorem~\ref{newTheoremBasis}).  

We denote by $\langle X\rangle$ the subgroup of the free group on $A$
generated by a
set of words $X$.
A submonoid $M$ of $A^*$ is called \emph{saturated} in $S$ 
if $M\cap S=\langle M\rangle\cap S$.
We recall the following result from
\cite{BertheDeFeliceDolceLeroyPerrinReutenauerRindone2013d}
(Theorem 6.2 referred to as the Saturation Theorem). 

\begin{theorem}\label{propositionHcapF}
  Let $S$ be an acyclic set. The submonoid generated by a bifix code 
included in $S$
is saturated in $S$.
\end{theorem}

Actually, by a second result
of~\cite{BertheDeFeliceDolceLeroyPerrinReutenauerRindone2013d} (Theorem 6.1
referred to as the Freeness Theorem), if $S$ is acyclic,
 any bifix code $X\subset S$ is free, which means
that it is a  basis of the subgroup $\langle
 X\rangle$.
We will not use this result here and thus we will prove
directly that if $S$ is a uniformly recurrent tree set,
any finite $S$-maximal bifix code is free.

Before proving Theorem~\ref{newTheoremBasis},
we list some related results. The first one is
the main result of~\cite{BerstelDeFelicePerrinReutenauerRindone2012}.

\begin{corollary}
A Sturmian set has the finite index basis property.
\end{corollary}
\begin{proof}
This follows from Theorem~\ref{newTheoremBasis} since a Sturmian set
is a uniformly recurrent tree set (Proposition~\ref{propositionSturmIsTree}).
\end{proof}

The following examples shows that Theorem~\ref{newTheoremBasis}
may be false for a set $S$ which does not satisfy some of the hypotheses.

The first example is a uniformly recurrent set which is not neutral.
\begin{example}
Let $S$ be the Chacon set (see Example~\ref{exampleChacon}).
We have seen that $S$ is not neutral and thus not a tree set.
The set $S\cap A^2=\{aa,ab,bc,ca,cb\}$ is an $S$-maximal bifix code of $S$-degree $2$.
It is not a basis since $ca(aa)^{-1}ab=cb$. Thus $S$ does
not satisfy the finite index basis property.
\end{example}
In the second example, the set is neutral but not a tree set and is
not
uniformly recurrent.
\begin{example}
Let $S$ be the set of Example~\ref{exampleNeutralNotTree}. It is
not a tree set (and it is not either uniformly recurrent). The set
$S\cap A^2$ is the same as in the Chacon set. Thus $S$ does
not satisfy the finite index basis property.
\end{example}

In the last example we have a uniformly recurrent set which is neutral
but not a
tree set.
\begin{example}\label{exampleJulien3}
Let $S$ be the set on the alphabet $B=\{1,2,3\}$ of
Example~\ref{exampleJulien}. We have seen that
$S$ is  neutral but not a tree set. 

  Let $X=S\cap B^2$. We have
$X=\{12,13,22,23,31\}$. The set $X$ is not a basis since
$13=12(22)^{-1}23$. Thus $S$ does
not satisfy the finite index basis property.
\end{example}

We close this section with a  converse of Theorem~\ref{newTheoremBasis}.

\begin{proposition}\label{propositionConverseFIB}
A biextendable set $S$ such that $S\cap A^n$ is a basis of the subgroup $\langle 
A^n\rangle$
for all $n\ge 1$ is a tree set.
\end{proposition}
\begin{proof}
Set $k=\Card(A)-1$.
Since $ A^n$ generates a subgroup of index $n$, the
hypothesis implies that $\Card(A^n\cap S)=kn+1$ for all $n\ge 1$.
Consider $w\in S$ and set $m=|w|$. The set
$X=AwA\cap S$ is  included in $Y=S\cap A^{m+2}$. Since
$Y$ is  a basis of a subgroup, $X\subset Y$ is
a basis of the subgroup  $\langle X\rangle$.

This implies that the graph
$G(w)$ is acyclic. Indeed, assume that $(a_1,b_1,\ldots,$
$a_p,b_p,a_1)$ is
a cycle in $G(w)$ with $p\ge 2$, $a_i\in L(w)$, $b_i\in R(w)$ for
$1\le i\le p$ and
$a_1\ne a_p$. Then
$a_1wb_1,a_2wb_1,\ldots,$ $a_pwb_p,a_1wb_p\in X$.
 But
\begin{displaymath}
a_1wb_1(a_2wb_1)^{-1}a_2wb_2\cdots a_pwb_p(a_1wb_p)^{-1}=1,
\end{displaymath}
contradicting the fact that $X$ is a basis.

Since $G(w)$ is an acyclic graph with
$\ell(w)+r(w)$ vertices and $e(w)$ edges, we have $e(w)\le
\ell(w)+r(w)-1$.
But then
\begin{eqnarray*}
\Card(A^{m+2}\cap S)=\sum_{w\in A^m\cap S}e(w)&\le&\sum_{w\in A^m\cap
  S}(\ell(w)+r(w)-1)\\
&\le& 2\Card(A^{m+1}\cap S)-\Card(A^m\cap S)\\
&\le&k(m+2)+1.
\end{eqnarray*}
Since $\Card(A^{m+2}\cap S)=k(m+2)+1$, we have $e(w)=\ell(w)+r(w)-1$
for all $w\in A^m$. This implies that $G(w)$ is a tree for all $w\in
S$.
Thus $S$ is a tree set.
\end{proof}

\begin{corollary}\label{corollaryConverseFiniteIndex}
A uniformly recurrent set which has the finite index basis property is a tree
set.
\end{corollary}
\begin{proof}
Let $S$ be a uniformly recurrent set having the finite index basis
property.
For any $n\ge 1$, the set $S\cap A^n$
is an $S$-maximal bifix code of $S$-degree $n$ (Example~\ref{exampleUniform}).
Thus it is a basis of a subgroup of index $n$.
Since it is included in the subgroup generated by $A^n$, which has index $n$, 
it is a basis
of this subgroup. This implies that $S$ is a tree set by 
Proposition~\ref{propositionConverseFIB}.
\end{proof}
%%%%%%%%%%%%%%%%%%%%%%%%%%

%%%%%%%%%%%%%%%%%%%%%
\subsection{Proof of the Finite Index Basis Theorem}
%%%%%%%%%%%%%%%%%%%%

Let $S$ be a  set of words. 
For $w\in S$, let
\begin{displaymath}
\Gamma_S(w)=\{x\in S\mid wx\in S\cap A^+w\}
\end{displaymath}
be  the set of \emph{right return words} to $w$. 
When $S$ is recurrent, the set $\Gamma_S(w)$ 
is nonempty. Let
\begin{displaymath}
\RR_S(w)=\Gamma_S(w)\setminus\Gamma_S(w) A^+
\end{displaymath}
be  the 
set of \emph{first right  return words}.

The proof of Theorem~\ref{newTheoremBasis} uses several other
results, among which Theorem~\ref{propositionHcapF}
and the following result
from~\cite{BertheDeFeliceDolceLeroyPerrinReutenauerRindone2013d}
(Theorem 5.6).
\begin{theorem}\label{theoremJulien}
Let $S$ be a uniformly recurrent tree set containing the alphabet $A$. For
any $w\in S$, the set $\RR_S(w)$ is a basis of the free group on $A$.
\end{theorem}
\begin{proofof}{of Theorem~\ref{newTheoremBasis}}
    Assume first that $X$ is a finite $S$-maximal bifix code of $S$-degree
  $d$. Let $P$ be the set of proper prefixes of $X$. Let $H$ be
  the subgroup generated by $X$.

Let $u\in S$ be a word such that $\delta_X(u)=d$,
or, equivalently, which is not an internal factor of $X$. Let $Q$ be the set formed
of the $d$ suffixes of $u$ which are in $P$.

\begin{figure}[hbt]
\centering
\gasset{Nadjust=wh,AHnb=0}
\begin{picture}(50,15)(0,-5)
\node(uh)(20,10){}\node(endh)(60,10){}
\node(u)(0,5){}\node(v)(40,5){}\node(end)(60,5){}
\node(q)(10,0){}\node(qend)(40,0){}\node(r)(30,-5){}\node(endr)(60,-5){}
\drawedge(uh,endh){$u$}
\drawedge(u,v){$u$}\drawedge(v,end){$y$}\drawedge(q,qend){$q$}
\drawedge(r,endr){$r$}
\end{picture}
\caption{A word $y\in \RR_S(u)$.}\label{figGamma}
\end{figure}
Let us first show that the cosets $Hq$ for $q\in Q$ are
disjoint. Indeed, assume that $Hp\cap Hq\ne\emptyset$. It  implies $Hp=Hq$.
But any $p,q\in Q$ are comparable for the suffix order. 
Assuming that $q$ is longer than $p$, we have
$q=tp$ for some $t\in P$. Then $Hp=Hq$ implies $Ht=H$ and thus $t\in
H\cap S$. By Theorem~\ref{propositionHcapF}, since $S$ is acyclic,
 this implies $t\in X^*$
and thus $t=\varepsilon$. Thus $p=q$.

Denote by $F_A$ the free group on $A$. Let
\begin{displaymath}
  V=\{v\in F_A\mid Qv\subset HQ\}\,.
\end{displaymath}

For any $v\in V$ the map $p\mapsto q$ from $Q$ into itself
defined by $pv\in Hq$
is a permutation of $Q$. Indeed, suppose that for
$p,p'\in Q$, one has $pv,p'v\in Hq$ for some $q\in Q$. Then $qv^{-1}$ is
in $Hp\cap Hp'$ and thus $p=p'$ by the above argument.

The set $V$ is a subgroup of $F_A$. Indeed, $1\in V$. Next, let $v\in V$. Then
for any $q\in Q$, since $v$ defines a permutation of $Q$, there
is a $p\in Q$ such that $pv\in
Hq$. Then $q v^{-1} \in Hp$. This shows that $v^{-1}\in V$.
Next, if $v,w\in V$, then $Qvw\subset HQw\subset HQ$ and thus $vw\in
V$.

We show that the set $\RR_S(u)$ is contained in $V$. Indeed, let $q\in
Q$ and $y\in \RR_S(u)$. Since $q$ is a suffix of $u$, $qy$ is a suffix
of $uy$, and since $uy$ is in $S$ (by definition of $\RR_S(u)$),
also $qy$ is in $S$.  Since $X$ is an $S$-maximal bifix code,
it is an $S$-maximal prefix code and thus
it is right $S$-complete. This
implies that $qy$ is a prefix of a word in $X^*$ and thus
there is a word $r\in P$ such that $qy\in X^*r$. We verify that the
word $r$ is a suffix of $u$.  Since $y\in \RR_S(u)$, there is a word
$y'$ such that $uy=y'u$. Consequently, $r$ is a suffix of $y'u$, and
in fact the word $r$ is a suffix of $u$. Indeed, one has $|r|\le |u|$
since otherwise $u$ is in the set $I(X)$ of internal
factors of $X$, and this is not the case. Thus we
have $r\in Q$ (see Figure~\ref{figGamma}). Since $X^*\subset H$ and
$r\in Q$, we have $qy\in HQ$. Thus $y\in V$.

By Theorem~\ref{theoremJulien}, the group generated by $\RR_S(u)$ is the free
group  on
$A$.  Since $\RR_S(u)\subset V$, and since $V$ is a subgroup of
$F_A$, we have $V=F_A$. Thus $Qw\subset HQ$ for any $w\in
F_A$.  Since $1\in Q$, we have in particular $w\in HQ$.  Thus
$F_A=HQ$. Since $\Card(Q)=d$, and since the right cosets $Hq$ for
$q\in Q$ are pairwise disjoint, this shows that $H$ is a subgroup of
index $d$. Since $S$ is a recurrent tree set, it is neutral and thus,
by Theorem~\ref{theoremCardinality}, we have $\Card(X)=
d(\Card(A)-1)+1$. But  $H$ is a free group which, by Schreier's
Formula, is of rank $d(\Card(A)-1)+1$.
Since $X$ generates $H$, this implies that 
 $X$ is a basis of $H$.

%-----------------------

Assume conversely that the finite bifix code $X\subset S$ is a basis of the
group $H=\langle X\rangle$ and that $H$ has index $d$. Since $X$ is a
basis of $H$, by Schreier's Formula, we have $\Card(X)= (k-1)d+1$, where
$k=\Card(A)$. The case $k=1$ is straightforward; thus we assume
$k\ge2$. 
By
Theorem 4.4.3 in~\cite{BerstelDeFelicePerrinReutenauerRindone2012},
 if $S$ is a uniformly recurrent set,
any finite bifix code contained in $S$ is contained in a finite
$S$-maximal bifix code. Thus there is a finite $S$-maximal bifix
code $Y$ containing $X$. Let $e$ be the $S$-degree of $Y$. By the
first part of the proof, $Y$ is a basis of a subgroup $K$ of index $e$
of the free group on $A$.  In particular, it has $(k-1)e+1$ elements. Since
$X\subset Y$, we have $(k-1)d+1\le (k-1)e+1$ and thus $d\le e$. On the
other hand, since $H$ is included in $K$, $d$ is a multiple of $e$ and
thus $e\le d$. We conclude that $d=e$ and thus that $X=Y$.
%-------------------------------------------
\end{proofof}

\bibliographystyle{plain}
\bibliography{finiteIndexBasisProperty}

\end{document}